\documentclass[a4paper,UKenglish,cleveref, autoref, thm-restate]{lipics-v2021}
\title{Faster Monotone Min-Plus Product, Range Mode, and Single Source Replacement Paths}

\titlerunning{Faster Monotone Min-Plus Product, Range Mode, and SSRP}

\author{Yuzhou Gu}{Massachusetts Institute of Technology, United States}{yuzhougu@mit.edu}{}{}
\author{Adam Polak}{École Polytechnique Fédérale de Lausanne, Switzerland}{adam.polak@epfl.ch}{https://orcid.org/0000-0003-4925-774X}{Supported by the Swiss National Science Foundation within the project \emph{Lattice Algorithms and Integer Programming} (185030). Part of this work was done at Jagiellonian University, supported by Polish National Science Center grant 2017/27/N/ST6/01334}
\author{Virginia {Vassilevska Williams}}{Massachusetts Institute of Technology, United States}{virgi@mit.edu}{}{Supported by an NSF CAREER Award, NSF Grants CCF-1528078, CCF-1514339 and CCF-1909429, a BSF Grant BSF:2012338, a Google Research Fellowship and a Sloan Research Fellowship.}
\author{Yinzhan Xu}{Massachusetts Institute of Technology, United States}{xyzhan@mit.edu}{}{Supported by NSF Grant CCF-1528078.}

\authorrunning{Y.\,Gu, A.\,Polak, V.\,Vassilevska Williams, and Y.\,Xu}

\Copyright{Yuzhou Gu, Adam Polak, Virginia Vassilevska Williams, and Yinzhan Xu}



\category{} %optional, e.g. invited paper

\relatedversion{}

\acknowledgements{This project was started as part of an open problem session held alongside MIT subject 6.890 taught by the third author. The authors would like to acknowledge the rest of the participants in the open problems session including but not limited to  Angelos Pelecanos, Nicole Wein and Yuancheng Yu.}

\nolinenumbers %uncomment to disable line numbering

\hideLIPIcs  %uncomment to remove references to LIPIcs series (logo, DOI, ...), e.g. when preparing a pre-final version to be uploaded to arXiv or another public repository

\usepackage{float}
\usepackage{tikz}
\usetikzlibrary{decorations.markings,decorations.pathreplacing,arrows, automata, positioning, quotes, shapes,backgrounds, arrows.meta, shapes.multipart}
\tikzstyle{vertex}=[circle, draw, inner sep=0pt, minimum size=4pt, fill = black]
\newtheorem{problem}[theorem]{Problem}

\def \qed {\hfill $\Box$}

\def \eps {\varepsilon}

\newcommand{\ignore}[1]{}

\def \polylog { \text{\rm polylog} }

\def\tO{\tilde{O}}
\def\ti{\tilde}
\def\p{\prime}

\newcommand{\IGNORE}[1]{}

\DeclareMathOperator*{\argmin}{arg\,min}

\begin{document}

\maketitle

\begin{abstract}
One of the most basic graph problems, 
All-Pairs Shortest Paths (APSP) is known to be solvable in $n^{3-o(1)}$ time, and it is widely open whether it has an $O(n^{3-\epsilon})$ time algorithm for $\epsilon > 0$. To better understand APSP, one often strives to obtain subcubic time algorithms for structured instances of APSP and problems equivalent to it, such as the Min-Plus matrix product.

A natural structured version of Min-Plus product is Monotone Min-Plus product which has been studied in the context of the Batch Range Mode [SODA'20] and Dynamic Range Mode [ICALP'20] problems. 
This paper improves the known algorithms for Monotone Min-Plus Product and for Batch and Dynamic Range Mode, and establishes a connection between Monotone Min-Plus Product and the Single Source Replacement Paths (SSRP) problem on an $n$-vertex graph with potentially negative edge weights in $\{-M, \ldots, M\}$.

SSRP with positive integer edge weights bounded by $M$ can be solved in $\tilde{O}(Mn^\omega)$ time, whereas the prior fastest algorithm for graphs with possibly negative weights [FOCS'12] runs in $O(M^{0.7519} n^{2.5286})$ time, the current best running time for directed APSP with small integer weights. Using Monotone Min-Plus Product, we obtain an improved $O(M^{0.8043} n^{2.4957})$ time SSRP algorithm, showing that SSRP with constant negative integer weights is likely easier than directed unweighted APSP, a problem that is believed to require $n^{2.5-o(1)}$ time. 

	Complementing our algorithm for SSRP, we give a reduction from the Bounded-Difference Min-Plus Product problem studied by Bringmann et al.~[FOCS'16] to negative weight SSRP. This reduction shows that it might be difficult to obtain an $\tilde{O}(M n^{\omega})$ time algorithm for SSRP with negative weight edges, thus separating the problem from SSRP with only positive weight edges.
	
\end{abstract}

%\newpage

\section{Introduction}

The All-Pairs Shortest Paths problem (APSP) is one of the most well-known problems in graph algorithms. Following the classical Floyd–Warshall algorithm that solves APSP in $O(n^3)$ time in $n$-vertex edge-weighted graphs, a long list of papers have been dedicated to improving the APSP running time. The current best algorithm by Williams~\cite{williams2014faster} runs in $n^3/2^{\Theta(\sqrt{\log n})}$ time.
It is a big open problem whether a truly sub-cubic, $O(n^{3-\epsilon})$-time for $\epsilon > 0$, algorithm for APSP exists. In fact, the popular APSP hypothesis from Fine-Grained Complexity \cite{virgisurvey} asserts that this is not the case. 

In the Min-Plus Product problem, one is given two $n\times n$ integer matrices $A, B$ and is required to compute  an $n\times n$ matrix $C=A\star B$ such that $C_{i, j} = \min_k \{ A_{i,k}+B_{k, j} \}$. Fischer and Meyer \cite{fischermeyer} showed that  Min-Plus Product is equivalent to APSP,
in the sense that a $T(n)$ time algorithm for either of the problems immediately implies an $O(T(n))$ time algorithm for the other.
The APSP hypothesis
thus states
that computing the min-plus product of $n\times n$ integer matrices requires $n^{3-o(1)}$ time\footnote{In fine-grained complexity one needs to fix the model of computation for each hardness hypothesis, and the APSP hypothesis is typically stated for a word RAM with $O(\log n)$ bit words, which is the model the algorithms in our paper are in.}.

\subparagraph*{Structured Min-Plus Products.}
To better understand the complexity of APSP, much research focuses on improving the running time for Min-Plus Product when one or both of the matrices have some structure, with the hope that eventually all instances can be handled.
As a fundamental problem, Min-Plus Product can be used to solve many other problems. It turns out that in many cases a structured version of Min-Plus Product suffices \cite{VX20, BGSV19}.  Thus, studying structured instances of Min-Plus Product has the potential to speed up the running times for many applications.

Alon,  Galil and Margalit \cite{AlonGM97} first studied the Min-Plus Product of structured matrices. They showed, following ideas of Yuval~\cite{Yuval76}, that if all entries of two $n\times n$ matrices $A, B$ are integers in $\{-M, \ldots, M\} \cup \{\infty\}$, then one can compute the min-plus product of $A$ and $B$ in $\tilde{O}(Mn^\omega)$ time\footnote{Throughout the paper the $\tilde{O}$ notation hides subpolynomial factors.}, where $\omega \in [2,2.373)$ denotes the best possible exponent of square matrix multiplication~\cite{Vassilevska12,LeGall14,AVW21}.

Yuster~\cite{yusterdom} considered Min-Plus Product when one of the matrices has a small number of distinct entries in each row, generalizing \cite{AlonGM97}.
Bringmann et al.~\cite{BGSV19} studied Min-Plus Product of bounded-difference matrices, generalizing \cite{AlonGM97,yusterdom}. An integer matrix is called to have \textit{bounded differences} if all pairs of adjacent entries (both horizontally and vertically) differ by at most $O(1)$. Bringmann et al.~\cite{BGSV19} gave an
$O(n^{2.8244})$ time algorithm for computing the Min-Plus Product between two bounded-difference matrices. When $\omega = 2$, their algorithm runs in
$O(n^{2.7554})$ time. They also studied variants of this problem including the case when only one matrix is guaranteed to have bounded differences, and the bounded-differences are only in the rows or only in the columns.

Chan \cite{chan2010more} gave a truly sub-cubic time algorithm for the min-plus product between \textit{geometrically weighted} matrices using a geometric tool called the \textit{partition theorem}. Recently, Vassilevska Williams and Xu \cite{VX20} combined the approach of Bringmann et al.~\cite{BGSV19}  and a geometric data structure
to give a truly subcubic Min-Plus Product algorithm for integer matrices where one of the matrices has constant $O(1)$-approximate rank, further generalizing the results of \cite{BGSV19} and partially \cite{chan2010more}.

\subparagraph*{Our contribution.} In this work, we study the Min-Plus Product of a {\em monotone} integer matrix with an arbitrary\footnote{Throughout the paper we assume the entries of the matrices are $\polylog(n)$-bit integers or $\infty$ unless otherwise stated.} integer matrix. We defer the general definition of {\em Monotone Min-Plus Product} to Section~\ref{sec:prelim}. For now let us focus on an interesting special case: We are given an arbitrary $n\times n$ integer matrix $A$ and an $n\times n$ matrix $B$ whose entries are positive integers bounded by $O(n)$, and such that each row of $B$ is non-decreasing.

The above special case already subsumes the Min-Plus Product of bounded-difference matrices studied by Bringmann et al.: Suppose we are asked to compute the min-plus product of matrices $A$ and $B$ where $B$  has bounded differences. In other words, all pairs of adjacent entries (both horizontally and vertically) differ by at most some constant $M$. We can create a matrix $B'$ so that $B'_{k, j} = B_{k,j} - B_{1, 1} + jM$. It is easy to check that the matrix $B'$ satisfies the above-mentioned special case definition of monotone matrix. Thus, we can use our algorithm to compute the min-plus product $C'=A\star B'$. Then it is easy to recover $C=A\star B$ by setting $C_{i,j}=C'_{i,j}+B_{1,1}-jM$. Therefore, Monotone Min-Plus Product is more general than the Min-Plus Product of an arbitrary matrix with a bounded-difference matrix.

Monotone Min-Plus Product was first studied by Vassilevska Williams and Xu \cite{VX20} as a tool to give a fast algorithm for the \emph{Batch Range Mode} problem. In their work, the authors devise a black-box reduction from Monotone Min-Plus Product to their Min-Plus Product algorithm for matrices with a small $O(1)$-approximate rank. 
Their algorithm runs in $\tO(n^{(15+\omega)/6})=O(n^{2.8955})$ time for the above-mentioned special case. We improve and generalize their algorithm. Below is a special case of our main theorem, which will be introduced in Section~\ref{sec:monotone}.
\begin{restatable}{theorem}{ThmMonotone}
	The min-plus product $A\star B$ of two $n \times n$ matrices where entries of $B$ are non-negative integers bounded by $O(n)$ and each row of $B$ is non-decreasing can be computed deterministically in $\tO(n^{\frac{12+\omega}{5}})$ time.
	Using the current best bound on fast rectangular matrix multiplication  the running time improves to $O(n^{2.8653})$.
\end{restatable}
\noindent If $\omega=2$, our improvement is from $\tO(n^{17/6}) \le O(n^{2.8334})$ time to $\tO(n^{14/5}) = \tO(n^{2.8})$ time.
We provide several interesting applications of our improved algorithm for Monotone Min-Plus Product.

\subsection{Applications}

\subparagraph*{Single Source Replacement Paths.}
The main contribution of this paper is establishing a relationship between Monotone Min-Plus Product and the Single-Source Replacement Paths (SSRP) problem. In the SSRP problem, one is given a directed edge-weighted graph $G$ and a source vertex $s$, and is asked to compute for each edge $e$, $d_G(s, v, e)$'s, the shortest path distances from $s$ to each vertex $v$ in $G\setminus \{e\}$. Note that the interesting case is when $e$ belongs to a shortest paths tree rooted at $s$, so that there are only $O(n^2)$ distances to report.

The trivial algorithm for SSRP runs in $\tilde{O}(n^3)$ time: For each edge $e$ on the shortest path tree rooted at $s$, run Dijkstra's algorithm on the graph with $e$ removed. Negative edge weights can be handled with Johnson's trick~\cite{Johnson1977}, without increasing the asymptotic complexity.
Vassilevska Williams and Williams \cite{williams2018subcubic} showed that APSP and SSRP are sub-cubically equivalent. Hence,
assuming the APSP hypothesis, there is no $O(n^{3-\epsilon})$ time algorithm for SSRP in graphs with arbitrary integer weights, for any $\epsilon > 0$. There seems to be little hope to improve upon the trivial algorithm for the general case.

Grandoni and Vassilevska Williams~\cite{GVW2019} studied SSRP in graphs with integer edge weights of small absolute value. They gave an algorithm that solves SSRP in directed $n$-vertex graphs with edge weights in $\{-M, \ldots, M\}$ in $\tilde{O}(M^{\frac{1}{4-\omega}}n^{2+\frac{1}{4-\omega}})$ time, which would become $\tO(M^{0.5}n^{2.5})$ if $\omega=2$. For positive weights only, they reduce the runtime to $\tO(Mn^\omega)$.

Let us consider the special case $M=1$. Here, the algorithms of \cite{GVW2019} solve SSRP with positive weights $1$ in $\tilde{O}(n^\omega)$ time, while the $\tO(n^{2+\frac{1}{4-\omega}})$ runtime for SSRP with weights in $\{-1,0,1\}$ is the same as the runtime for APSP with weights in $\{-1,0,1\}$. 

As APSP in graphs with arbitrary integer weights is fine-grained equivalent to SSRP with arbitrary integer weights \cite{williams2018subcubic}, it is possible that APSP with weights in $\{-1,0,1\}$ could be fine-grained equivalent to SSRP with weights in $\{-1,0,1\}$. 

It is believed that APSP in directed graphs with weights in $\{-1,0,1\}$ (and even for unweighted graphs) requires $n^{2.5-o(1)}$ time \cite{Lincoln20,WX20FOCS}, as the best known algorithm by Zwick \cite{Zwick02} would run in $\Omega(n^{2.5})$ time even if $\omega=2$. 
As SSRP with small {\em positive} weights is in $O(n^{2.5-\eps})$ time for $\eps>0$ \cite{GVW2019}, it is likely not fine-grained equivalent to directed unweighted  APSP. Beating the APSP runtime for SSRP with {\em negative} weights is an open problem.

This leads to the following interesting questions.

\begin{center}
\emph{(1) Is SSRP with negative weights inherently harder than SSRP with only positive weights?\\ (2) Or, is it possible to improve the running time of SSRP with negative weights, possibly below $n^{2.5}$, thus showing that it is likely not as hard as directed unweighted APSP?}
\end{center}

Quite surprisingly, we give positive answers to both of these questions.
First, we improve over the running time of \cite{GVW2019} for negative weights.

\begin{restatable}{theorem}{ThmSSRPUpperBound}\label{thm:SSRP_upper_bound}
There is a randomized algorithm that solves SSRP in a directed $n$-vertex graph with edge weights in $\{-M, \ldots, M\}$ in $\tO(M^{\frac{5}{17-4\omega}}n^{\frac{36-7\omega}{17-4\omega}})$ time, with high probability. Using the current best bound on fast rectangular matrix multiplication  the running time improves to $O(M^{0.8043}n^{2.4957})$.
\end{restatable}

Notably, when $M$ is small enough, the running time $O(M^{0.8043}n^{2.4957})$ is polynomially faster than $n^{2.5}$, and hence faster than the best known running time of APSP in directed unweighted graphs which is $\Omega(n^{2.5})$ even if $\omega=2$. This answers our question (2) above.
If $\omega = 2$, our running time for SSRP with negative weights is $\tO(M^{5/9}n^{22/9}) \le O(M^{0.556}n^{2.445})$.

APSP in directed graphs with edge weights in $\{-1,0,1\}$ is one of long list of so-called \emph{intermediate} graph and matrix problems~\cite{Lincoln20,WX20FOCS}, whose running time lies between $\tO(n^\omega)$ and $\tO(n^3)$ and becomes $\tO(n^{2.5})$ when $\omega=2$.
Our result shows that SSRP with bounded negative integer weights is not an intermediate problem. We remark that recently Grandoni et al.~\cite{Grandoni21} showed that another (ex-)candidate intermediate problem, All-Pairs LCA in DAGs, can actually be solved faster than $O(n^{2.5})$ time.

We prove Theorem~\ref{thm:SSRP_upper_bound} by improving the runtime of the so-called \emph{subpath problem}, which is the bottleneck in the algorithm of \cite{GVW2019}. Grandoni and Vassilevska Williams solve it by reducing to APSP in directed graphs with edge weights in $\{-M, \ldots, M\}$, and applying Zwick's APSP algorithm~\cite{Zwick02}. We show that the APSP computation can be rearranged so that certain min-plus products that appear throughout involve monotone matrices. 

Next, we identify an obstacle to obtaining a $\tilde{O}(Mn^\omega)$ time algorithm for SSRP with negative weights, addressing our question (1).

\begin{restatable}{theorem}{ThmSSRPLowerBound}\label{thm:SSRP_lower_bound}
If there exists a $T(n)$ time algorithm for SSRP in $n$-vertex graphs with edge weights in $\{-1, 0, 1\}$, then there exists an $O(T(n) \sqrt{n})$ time algorithm for the Bounded-Difference Min-Plus Product of $n \times n$ matrices.
\end{restatable}
Theorem~\ref{thm:SSRP_lower_bound} gives the following argument why SSRP with negative weights might be hard. The current best algorithm for Bounded-Difference Min-Plus Product runs in $O(n^{2.7554})$ time even if $\omega = 2$. If SSRP with weights $\{-1, 0, 1\}$ could be solved in $\tilde{O}(n^2)$ time (when $\omega=2$), then Bounded-Difference Min-Plus Product could be solved in $\tilde{O}(n^{2.5})$ time, which would be a breakthrough in structured Min-Plus Product algorithms.

Recently replacement paths problems have received increased attention~\cite{ChechikC19,AlonCC19,ChechikM20,ChechikN20}. None of these works is directly related to ours, because they focus either on the $s$-$t$ Replacement Paths problem (with both source and target nodes fixed), or on combinatorial algorithms (i.e.~without fast matrix multiplication) for sparse graphs.

\subparagraph*{Range Mode.}
 Given an array $a$ of elements, a range mode query asks for the most frequent element in a contiguous interval of $a$. In the Batch Range Mode problem the array $a$ is fixed and all range mode queries are given in advance. In the Dynamic Range Mode problem one starts with an empty array and has to support insertions and deletions, and handle queries in an online fashion.

Vassilevska Williams and Xu \cite{VX20} were the first to use structured Min-Plus Product in range mode algorithms. They reduced Batch Range Mode to a Min-Plus Product instance where both matrices have some monotone structures. Their techniques give an $O(n^{1.4854})$ time algorithm for Batch Range Mode on an array of size $n$ and $n$ queries. Since we improve over their Monotone Min-Plus Product algorithm, we naturally obtain a faster Batch Range Mode algorithm.

\begin{restatable}{theorem}{ThmBatch}\label{thm:batch-range-mode}
	The Batch Range Mode problem can be solved deterministically in time $\tO(n^\frac{21+2\omega}{15+\omega})$. Using the current best bound on fast rectangular matrix multiplication  the running time improves to
	$O(n^{1.4805})$.
\end{restatable}

There are multiple algorithms that solve Dynamic Range Mode in $\tilde{O}(n^{2/3})$ time per update and query on an array of size bounded by $n$ \cite{Chan14, zhms2018}. Recently, Sandlund and Xu \cite{SX20} improved both update and query time to $O(n^{0.6560})$ by using a so-called \textit{Min-Plus-Query-Witness} problem. During the preprocessing phase of the Min-Plus-Query-Witness problem, one is given two matrices $A, B$. During the query phase, given two indices $i, j$ and a set $S$, one is asked to compute $\argmin_{k \not \in S} \{A_{i,k}+B_{k,j}\}$, where the set $S$ can be viewed as the set of elements recently deleted in the array. In the Min-Plus-Query-Witness instances reduced from Dynamic Range Mode, the matrices $A, B$ have the monotone property, so our techniques for Monotone Min-Plus Product can also apply to these Min-Plus-Query-Witness instances, leading to a faster Dynamic Range Mode algorithm.

\begin{restatable}{theorem}{ThmDynamic}\label{thm:dynamic-range-mode}
	The Dynamic Range Mode problem can be solved deterministically in $\tO(n^{\frac{\omega+9}{\omega+15}})$ worst-case time per query with $\tO(n^{\frac{3\omega+39}{2\omega + 30}})$ space. Using the current best bound on fast rectangular matrix multiplication improves  the running time  to $O(n^{0.6524})$ and the space complexity to $O(n^{1.3262})$.
\end{restatable}

\subsection{Overview of the Monotone Min-Plus Product Algorithm}\label{subsec:mmp-overview}
Our improvement is achieved by extending Vassilevska Williams and Xu's~\cite{VX20} framework so that it can handle the more general monotone matrices.
The algorithm has three phases.
Say we would like to compute the min-plus product $C = A\star B$, where $B$ is a monotone matrix.

In Phase 1, we compute a matrix $\ti C$ which is close in $\ell_\infty$ norm to the desired output $C$.
This can be done by, e.g., computing the min-plus product of $\lfloor \frac{A}{W} \rfloor$ and $\lfloor \frac{B}{W} \rfloor$,
the downscaled versions of $A$ and $B$, for  some small parameter $W$. 

In Phase 2, we repeatedly sample columns of $B$, and create new matrices $A^r$ and $B^r$
(for $r=1,2,\ldots$) whose entries are simple linear combinations of $A$, $B$ and $\ti C$.
We replace large-magnitude entries of $A$ with $\infty$, so that all finite entries of $A^r$ are of
small absolute values, and that the min-plus product $A^r \star B^r$ can be computed efficiently.
The results are collected in a way such that
by the end of Phase 2,
we have found, for every pair $(i,j)$,
\[\min_k\{A_{i,k} + B_{k,j} : A^r_{i,k}\ne \infty\text{ for some }r\}.\]

In Phase 3, we deal with $(i,k)$ such that $A^r_{i,k}=\infty$ for all $r$. Such $(i,k)$ are called \emph{uncovered}.
It can be shown that the number of \emph{relevant} triples $(i,k,j)$ with $(i,k)$ uncovered is small,
and we can afford enumerating all such triples.
The enumeration is done by performing a witness-listing version of min-plus product of the downscaled versions
of $A$ and $B$.

We base on the fact
that, when a monotone matrix $B$ has very small entries, the number of \emph{changes},
i.e. pairs $(k,j)$ for which $B_{k,j} \ne B_{k,j+1}$, can be upper-bounded.
Then for each fixed $i$ we let $j$ iterate through its range, and we maintain the set $\{A_{i,k} + B_{k,j}\}$
during the iteration.
The total number of updates to the set is exactly the number of changes.
This gives us an efficient way to compute min-plus product (and its witness-listing version)
of the downscaled matrices,
and leads to an efficient running time for Phases 1 and 3.

What makes our improvement possible is that we focus directly on the structure of monotone matrices, instead of going through a lossy black-box reduction to the Min-Plus Product of bounded-difference matrices, like the previous work~\cite{VX20} did.

\section{Preliminaries}\label{sec:prelim}

\begin{definition}[Rectangular Matrix Multiplication Exponent]\label{defn:w}
	Let $\alpha$, $\beta$, $\gamma$ be non-negative real numbers.
	Define $\omega(\alpha,\beta,\gamma)$
	to be the smallest number
	such that the product of an $n^\alpha \times n^\beta$ matrix by an $n^\beta \times n^\gamma$ matrix	can be computed in
	$\tO(n^{\omega(\alpha,\beta,\gamma)})$ time.
\end{definition}

\begin{definition}\label{defn:g}
	Let $\alpha$, $\beta$, $\gamma$, $\theta$ be non-negative real numbers.
	Define $g(\alpha,\beta,\gamma,\theta)$
	to be the smallest number such that the min-plus product of an $n^\alpha \times n^\beta$
	matrix whose entries are in $\{-n^\theta, \ldots, n^\theta\}\cup\{\infty\}$
	by an arbitrary $n^\beta \times n^\gamma$ matrix
 	can be computed in $\tO(n^{g(\alpha,\beta,\gamma,\theta)})$ time.
\end{definition}

\begin{definition}[Bounded-Difference Matrix]
An $n \times m$ matrix $A$ is called a bounded-difference matrix if $|A_{i,j}-A_{i, j + 1}| \le 1$ for every $1 \le i \le n, 1 \le j < m$ and $|A_{i,j}-A_{i+1, j}| \le 1$ for every $1 \le i < n, 1 \le j \le m$.
\end{definition}

\begin{problem}[Bounded-Difference Min-Plus Product]
Given two bounded-difference integer matrices $A$ and $B$, compute $A \star B$.
\end{problem}

\begin{definition}[Monotone Matrix]\label{defn:monotone}
	An $n^\beta\times n^\gamma$ matrix $B$ is called monotone if for every $k\in [n^\beta]$, $B_{k,j}$ is non-decreasing in $j\in [n^\gamma]$.
	For a monotone matrix $B$, we define its total range as
	$\sum_{k\in [n^\beta]} (\max_{j \in [n^\gamma], B_{k, j} \ne \infty} B_{k,j} - B_{k,1} + 1).$
\end{definition}

\begin{problem}[Monotone Min-Plus Product]
    Given an $n^\alpha \times n^\beta$ matrix $A$ and an $n^\beta \times n^\gamma$ matrix $B$
    where $B$ is monotone and has total range $O(n^{\beta+\eta})$,
	where $\alpha$, $\beta$, $\gamma$, $\eta$ are non-negative real numbers, compute the min-plus product $A\star B$.
\end{problem}

\begin{definition}\label{defn:m}
	Define $m(\alpha,\beta,\gamma,\eta)$ to be the
	smallest number
	such that Monotone Min-Plus Product with parameters $\alpha$, $\beta$, $\gamma$, $\eta$, 
	can be computed in $\tO(n^{m(\alpha,\beta,\gamma,\eta)})$ time.
\end{definition}

In our applications, we only need the case $\alpha = \gamma$.
Because $m(c \alpha, c \beta, c \gamma, c \eta) = c m(\alpha, \beta, \gamma, \eta)$ for any $c\ge 0$, it suffices to consider only the case $\alpha=\gamma=1$.

\subsection{Upper bounds for \texorpdfstring{$g$}{g}}
We prove some useful upper bound for the function $g(\cdot,\cdot,\cdot,\cdot)$ introduced in Definition \ref{defn:g}.
We use the following bound from~\cite{CVX}, which is a straightforward generalization of Theorem~1.2 from~\cite{VX20} to the case of rectangular matrices.
\begin{lemma}
For any non-negative real numbers $\alpha, \beta, \gamma, \theta$, 
\begin{equation}
	g(\alpha,\beta,\gamma,\theta) \le \min_{0\le \delta\le \beta} \max\{ \omega(\alpha,\beta,\beta+\gamma-\delta) + \theta, \delta + \alpha+\gamma\}.\label{eqn:g-bound-rmm}
\end{equation}
\end{lemma}

We only need the following special cases of \eqref{eqn:g-bound-rmm}.
\begin{corollary}\label{cor:g}

For any non-negative real numbers $\beta$ and $\theta$, 
	\begin{align}
		g(1,\beta,1,\theta) &\le  \frac 12(2+\beta+\omega(1,\beta,1)+\theta), \label{eqn:g-1b1-bound}\\
		g(1,1,1,\theta) &\le \min_{0\le \delta\le 1} \max\{\omega(1,1,2-\delta)+\theta, 2+\delta\}.\label{eqn:g-111-bound}
	\end{align}
\end{corollary}

\begin{proof}
	Consider the term $\omega(\alpha,\beta,\beta+\gamma-\theta)$ in \eqref{eqn:g-bound-rmm}.
	By splitting matrix $B$ into $n^{\beta-\delta}$ matrices along its second dimension and computing $n^{\beta-\delta}$ independent instances of matrix multiplications, we get
	\begin{equation}
		\omega(\alpha,\beta,\beta+\gamma-\delta) \le \omega(\alpha,\beta,\gamma) + \beta-\delta.\label{eqn:omega-add}
	\end{equation}
	We plug \eqref{eqn:omega-add} into \eqref{eqn:g-bound-rmm}, and take
	$\delta = \min\{\beta, \frac 12 (\omega(\alpha,\beta,\gamma) + \beta+\theta-\alpha-\gamma)\},$
	to get
	\begin{equation}
		g(\alpha,\beta,\gamma,\theta) \le \max\{\omega(\alpha,\beta,\gamma)+\theta, \frac 12 (\alpha+\beta+\gamma+\omega(\alpha,\beta,\gamma)+\theta)\}.\label{eqn:g-rmm-proof}
	\end{equation}
	However, if
	$\omega(\alpha,\beta,\gamma)+\theta\ge \frac 12 (\alpha+\beta+\gamma+\omega(\alpha,\beta,\gamma)+\theta),$
	then both sides of \eqref{eqn:g-rmm-proof} are at least $\alpha+\beta+\gamma$,
	and we can compute the min-plus product in $O(n^{\alpha+\beta+\gamma})$ time using a trivial algorithm.
	Therefore we get \eqref{eqn:g-1b1-bound}.

	\eqref{eqn:g-111-bound} is a simple substitution $\alpha=\beta=\gamma=1$ to \eqref{eqn:g-bound-rmm}.
\end{proof}
\begin{remark}\label{rmk:rmm}
	In the rectangular case $\alpha=\gamma=1$, we use \eqref{eqn:omega-add} so that in the final expression \eqref{eqn:g-1b1-bound}
	we only need to deal with terms of the form $\omega(1, \beta, 1)$, whose value can be bounded by \cite{LU18}.
	We know of no handy upper bounds for $\omega(\alpha,\beta,\gamma)$ when all three parameters are distinct.

	In the square case $\alpha=\beta=\gamma=1$, we do not need to use the simplification \eqref{eqn:omega-add}.
	This is because the upper bound of $\omega(1, \beta, 1)$ in \cite{LU18} is by a bilinear algorithm.
	Thus by \cite{HP98}, the same upper bound works for $\omega(1, 1, \beta)$.
\end{remark}

\section{Monotone Min-Plus Product}\label{sec:monotone}
\begin{theorem}\label{thm:mmp-main}
	The min-plus product of an $n \times n^\beta$ matrix $A$ and an $n^\beta \times n$ matrix $B$ where $B$ is monotone with total range $O(n^{\beta+\eta})$
	can be deterministically computed, for any $\theta \in [0, \eta]$, in time
	$\tO(n^{\max\{1+\beta+\eta-\theta, \frac 12(2+\beta+g(1,\beta,1,\theta))\}}).$
    % by a deterministic algorithm.
	In other words,
	\[
		m(1, \beta, 1, \eta)\le\min_{0\le \theta\le \eta} \max\{1+\beta+\eta-\theta, \frac {1}{2}(2+\beta+g(1,\beta,1,\theta))\}.
 	\]
\end{theorem}

\begin{proof}
Proof of Theorem~\ref{thm:mmp-main} follows the three-phase framework of \cite{BGSV19,VX20}, which we have briefly described in Section \ref{subsec:mmp-overview}.
Here we present the full algorithm.

\subparagraph*{Phase 1.}
Let $\theta \in [0, \eta]$ be a parameter, and let $W = \lfloor n^\theta \rfloor$.
We define two matrices $\ti A$ and $\ti B$ as $\ti A_{i,k} = \lfloor \frac {A_{i,k}}W\rfloor$
and $\ti B_{k,j} = \lfloor \frac {B_{i,k}}W\rfloor$.
We compute the min-plus product $\ti A \star \ti B$ and let $\ti C_{i,j} = (\ti A \star \ti B)_{i,j} W$.
Then $||\ti C - C||_\infty \le 2 W$.

We compute $\ti A\star \ti B$ using the following lemma, which is a simple algorithm that works fast when the total range is very small.
\begin{lemma}\label{lemma:mmp-simple}
	$m(\alpha,\beta,\gamma,\eta) \le \max\{\alpha+\gamma, \beta+\gamma, \alpha+\beta+\eta\}.$
\end{lemma}
\begin{proof}
	Say we would like to compute $C = A\star B$.
	For a fixed row $i\in [n^\alpha]$, we iterate through columns $j\in [n^\gamma]$, maintaining the multi-set $\{A_{i,k} + B_{k,j} : k\in [n^\beta]\}$.

	Each time $j$ increases, we need to update the multi-set for those $k$ where $B_{k,j}\ne B_{k,j-1}$.
	The total number of $(k,j)$ satisfying $B_{k,j}\ne B_{k,j-1}$ is $O(n^{\beta+\eta})$ by monotonicity and the bound on total range.

	For each $i\in [n^\alpha]$, we need to make $O(n^{\beta+\eta})$ updates and $O(n^{\gamma})$ queries
	for the minimum number in the multi-set.
	We can use a balanced BST to maintain the multi-set so that each update and query costs $\tO(1)$ time.
	The total running time is $\tO(n^{\max\{\alpha+\gamma, \beta+\gamma, \alpha+\beta+\eta\}})$.
\end{proof}

Total range of $\ti B$ is $O(n^{\beta+\eta-\theta})$, so running time of Phase 1 is $\tO(n^{\max\{2, 1+\beta+\eta-\theta\}})$.

\subparagraph*{Phase 2.}
In Phase 2 we compute a matrix $\hat C$ which upper bounds $C = A\star B$ and agrees with it on most entries.
Initially, let $\hat C_{i,j} \leftarrow \infty$ for all $i,j\in [n]$.

Phase 2 consists of $(10+\beta)n^\rho \log n$ \emph{rounds}, for a parameter $\rho\ge 0$ to be chosen later. In the $r$-th round, we choose $j^r\in [n]$ uniformly at random\footnote{For simplicity of presentation, we use randomness here. The derandomization is deferred to Section~\ref{sec:derand}.}
Define matrix $A^r$ and $B^r$ as
\begin{align*}
&    A^r_{i,k} = 
\left\{
\begin{array}{ll}
A_{i,k} + B_{k,j^r} - \ti C_{i,j^r} & \text{if } A_{i,k} + B_{k,j^r} - \ti C_{i,j^r} \le 3W \text{ and } A^{r^\p}_{i,k}=\infty\forall r^\p<r,\\
\infty & \text{otherwise,}
\end{array}
\right.
\\
&B^r_{k,j} = 
\left\{
\begin{array}{ll}
B_{k,j} - B_{k,j^r} & \text{if } B_{k, j^r} \ne \infty,\\
0 & \text{otherwise.}
\end{array}
\right.
\end{align*}

We compute $C^r = A^r \star B^r$ using Corollary~\ref{cor:g} because $A^r$ has bounded entries.
Finally, for all $i, j\in [n]$, we make the update $\hat C_{i,j}\leftarrow \min\{\hat C_{i,j}, C^r_{i,j} + \ti C_{i,j^r}\}$.

In other words, in the end we have $\hat C_{i,j} = \min_r \{C^r_{i,j} + \ti C_{i,j^r}\}$.
If $A^r_{i,k}\ne \infty$, then for all $j$, we have
$\hat C_{i,j} \le C^r_{i,j} + \ti C_{i,j^r} \le A^r_{i,k} + B^r_{i,k} + \ti C_{i,j^r} = A_{i,k} + B_{k,j}.$
Thus in this case we have effectively updated $\hat C_{i,j}$'s using $A_{i,k} + B_{k,j}$ for all $j$.

% So the triple $(i,k,j)$ is ``covered'' for all $j\in [n]$.
Following \cite{VX20}, we make the following definitions: a triple $(i,k,j)\in [n]\times [n^\beta] \times [n]$ is  \emph{strongly relevant}, if $A_{i,k} + B_{k,j} = C_{i,j}$;
	\emph{weakly relevant}, if $A_{i,k} + B_{k,j} - \ti C_{i,j} \le 3W$;
	\emph{covered}, if $A^r_{i,k} \ne \infty$ for some $r$;
	\emph{uncovered}, if it is not covered.
We use the following lemma from \cite{VX20}.
\begin{lemma}[{\cite[Lemma 4.3]{VX20}}]\label{lemma:few-weak}
	With probability $1-n^{-9}$, the number of triples that are weakly relevant and uncovered is at most $n^{2+\beta-\rho}$.
\end{lemma}
\begin{proof}
    Fix some pair of $(i, k)$. If the number of $j$ such that $(i, k, j)$ is weakly relevant is at least $n^{1-\rho}$, then with probability at least $1-(1-n^{-\rho})^{(10+\beta)n^\rho \log n} \ge 1-n^{-10-\beta}$, we will sample a $j^r$ such that $(i, k, j^r)$ is weakly relevant. If so, $A^r_{i, k} \ne \infty$ for some $r$ and thus $(i, k, j)$ will be covered for all $j$. 
Therefore, with probability at least $1-n^{-9}$, all $(i, k)$ that are in at least $n^{1-\rho}$ weakly relevant triples will be covered. The number of remaining weakly relevant triples is at most $n^{2+\beta-\rho}$.
\end{proof}

Each round costs $\tO(n^{g(1,\beta,1,\theta)})$ time.
Phase 2 costs $\tO(n^{\rho + g(1,\beta,1,\theta)})$ time in total.

\subparagraph*{Phase 3.}
We define a triple $(i,k,j)$ to be \emph{moderately relevant} if $\ti A_{i,k} + \ti B_{k,j} \le (\ti A \star \ti B)_{i,j}+1$.
In Phase 3, we enumerate over moderately relevant and uncovered triples to complete matrix $C$.

\begin{lemma}\label{lemma:strong-moderately}
    Every strongly relevant triple is also moderately relevant.
\end{lemma}
\begin{proof}
	Suppose $(i,k,j)$ is strongly relevant.
	If $\ti A_{i,k^\p} + \ti B_{k^\p,j} = (\ti A\star \ti B)_{i,j}$ for some $k'$, then
	because $A_{i,k}+B_{k,j} \le A_{i,k^\p} + B_{k^\p,j}$,
	we have
	\[\ti A_{i,k} + \ti B_{k,j} \le \ti A_{i,k^\p} + \ti B_{k^\p,j} + 1 = (\ti A \star \ti B)_{i,j}+1.\]
	Hence $(i,k,j)$ is moderately relevant.
\end{proof}
\begin{lemma}\label{lemma:moderately-weak}
    Every moderately relevant triple is also weakly relevant.
\end{lemma}
\begin{proof}
	Suppose $(i,k,j)$ is moderately relevant. Then
	\begin{align*}
		A_{i,k}+B_{k,j}-\ti C_{i,j} &\le (\ti A_{i,k}+1) W + (\ti B_{k,j}+1) W - (\ti A \star \ti B)_{i,j}W\\
		&\le (\ti A_{i,k} + \ti B_{k,j} - (\ti A \star \ti B)_{i,j}) W + 2W
		\le 3W.
	\end{align*}
	Hence $(i,k,j)$ is weakly relevant.
\end{proof}

By Lemma~\ref{lemma:strong-moderately}, it suffices to enumerate over moderately relevant and uncovered triples
to recover all of $C$.
By Lemmas~\ref{lemma:few-weak} and~\ref{lemma:moderately-weak},
the number of moderately relevant and uncovered triples is at most $O(n^{2+\beta-\rho})$, with high probability.
\begin{lemma}\label{lemma:mmp-phase-3}
	With high probability, it takes time $\tO(n^{\max\{2, 1+\beta+\eta-\theta, 2+\beta-\rho\}})$ to enumerate all moderately relevant and uncovered triples.
\end{lemma}
\begin{proof}
	Define matrix $\check A$ as
	$\check A_{i,k} = \ti A_{i,k}$ if $(i,k)$ is uncovered; and $\check A_{i,k} = \infty$ otherwise.

	We proceed on computing $\check A \star \ti B$ in a way similar to Lemma~\ref{lemma:mmp-simple} from Phase 1.
	For each row $i$, we maintain the set $\{(\check A_{i,k} + \ti B_{k,j}, k) : k\in [n^\beta]\}$ as $j$ iterates over $[n]$.
	Each time $j$ increases, we need to update the multi-set for those $k$ where $\tilde{B}_{k,j}\ne \tilde{B}_{k,j-1}$.
	The total number of $(k,j)$ satisfying $\tilde{B}_{k,j}\ne \tilde{B}_{k,j-1}$ is $O(n^{\beta+\eta-\theta})$.
	
	For each $(i,j)$, we enumerate the elements in the multi-set in the increasing order,
	and stop as soon as we observe a $k$ where $\check A_{i,k} + \ti B_{k,j} > (\ti A \star \ti B)_{i,j}+1$. Therefore we enumerate exactly the moderately relevant uncovered triples.
	The running time is the running time from Lemma~\ref{lemma:mmp-simple}, plus the
	number of triples emitted, which, with high probability, is at most $O(n^{2+\beta-\rho})$, by Lemma~\ref{lemma:few-weak}.
\end{proof}

Phase 3 runs in time $\tO(n^{\max\{2, 1+\beta+\eta-\theta, 2+\beta-\rho\}})$.

\subparagraph*{Summary.}
Overall running time of our algorithm is
$\tO(n^{\max\{2, 1+\beta+\eta-\theta, \rho + g(1,\beta,1,\theta), 2+\beta-\rho\}})$.
Note that $g(1,\beta,1,\theta) \le 2+\beta$.
So we can take $\rho = \frac 12 (2+\beta - g(1, \beta, 1, \theta))$.
Also note that $\rho + g(1,\beta,1,\theta) \ge 2$, so the $2$ in the $\max$ expression can be ignored.
In the end we get
$\tO(n^{\max\{1+\beta+\eta-\theta, \frac 12 (2+\beta + g(1, \beta, 1, \theta))\}})$
as claimed.
\end{proof}

As a benchmark, let us consider the case $\alpha=\beta=\gamma=\eta=1$.
\ThmMonotone*
\begin{proof}
	In this case, Theorem \ref{thm:mmp-main} simplifies (via \eqref{eqn:g-111-bound}) to
	\begin{align}
		m(1,1,1,1) &\le \min_{0\le \theta \le 1} \max\{3-\theta, \frac 12 (3+g(1,1,1,\theta))\} \nonumber \\
		&\le \min_{0\le \theta \le 1} \max\{3-\theta, \frac 12 (3+\min_{0\le \delta\le 1} \max \{\omega(1,1,2-\delta)+\theta, 2+\delta\})\}. \label{eqn:m111}
	\end{align}
    
    Without using rectangular matrix multiplication, we can use $\omega(1,1,2-\delta) \le 1-\delta + \omega$ and take $\theta = \frac{3-\omega}{5}$ and $\delta = \frac{2\omega-1}{5}$, so \eqref{eqn:m111} takes value $\frac{12+\omega}5$.
    
    Using the rectangular matrix multiplication upper bounds in \cite{LU18} (see also Remark \ref{rmk:rmm}), we find that when
	$\theta = 0.1348$, $\delta = 0.7305$,
	expression \eqref{eqn:m111} takes value $\le 2.8653$.
\end{proof}

\subsection{Derandomization for Theorem \ref{thm:mmp-main}}
\label{sec:derand}
In this section we describe how to derandomize Phase 2 in the algorithm for Theorem \ref{thm:mmp-main},
achieving a deterministic algorithm with the same running time.

In the derandomized algorithm, we still choose $j^r$ and compute the min-plus product $A^r \star B^r$,
where $A^r$ and $B^r$ are defined in the same way as in the original Phase 2.
Instead of choosing $j^r$ randomly, we pick $j^r$ with a large enough number of pairs $(i,k)$
such that $(i,k,j^r)$ is uncovered and moderately relevant.

More specifically, we choose a number $\rho \ge 0$ and iterate through $j\in [n]$.
In each iteration $j$, we do the following.

\begin{itemize}
	\item Maintain for each $i\in [n]$ the set
	\[\{((\ti A_{i,k} + \ti B_{k,j}, k) : \text{$(i,k)$ is uncovered yet}\}\]
	using a balanced BST that supports querying the number of elements smaller than a given value in $\tO(1)$ time per query.
	\item Count the total number of $(i,k)$ such that
	$(i,k,j)$ is moderately relevant and uncovered yet by making one query for each $i$.
	\item If this number is smaller than $n^{1+\beta-\rho}$, do nothing and go to the next iteration.
	\item Otherwise, let this $j$ be the new $j^r$.
	Compute $A^r \star B^r$ and update $\hat C$ accordingly.
	\item Remove $(\ti A_{i,k} + \ti A_{k,j}, k)$
	from $i$-th set for all $(i,k)$ covered in this iteration.
\end{itemize}

After this process, for each $j$, the number of moderately relevant and uncovered $(i,k,j)$ is at most $n^{1+\beta-\rho}$,
because otherwise this $j$ would have been chosen as a $j^r$ and all those $(i,k)$ would have been covered.
Therefore the total number of moderately and uncovered triples is at most $n^{2+\beta-\rho}$,
and Phase 3 can go through.

Because for each $r$, we cover at least $n^{1+\beta-\rho}$ new pairs $(i,k)$,
the total number of different $r$'s is at most $n^\rho$.
The running time for each $r$ is $\tO(n^{g(1,\beta,1,\theta)})$ and the running time for maintaining the balanced BSTs is $\tO(n^{\max\{2, 1+\beta+\eta-\theta\}})$.
Thus, the total running time for derandomized Phase 2 is $\tO(n^{\max\{\rho + g(1,\beta,1,\theta), 1+\beta+\eta-\theta\}}),$
and the asymptotic running time for the whole algorithm stays the same.

\section{Single Source Replacement Paths}

We show our algorithm and lower bound for SSRP in this section.  We use $d_G(u, v)$ to denote the length of a shortest path from $u$ to $v$ in a graph $G$, and we use $d_G(u, v, e)$ as a shorthand for $d_{G\setminus\{e\}}(u, v)$. When it is clear from the context, we sometimes omit $G$.

\subsection{Algorithm}
In this section we present our improved algorithm for SSRP, proving Theorem~\ref{thm:SSRP_upper_bound}.

\ThmSSRPUpperBound*

To this end we improve the bottleneck in Grandoni-Vassilevska Williams algorithm~\cite{GVW2019}, hence let us begin with a high level overview of that algorithm. This is however just to give a context and intuition, and our formal proof of Theorem~\ref{thm:SSRP_upper_bound} follows from a black-box\footnote{See end of the proof of Lemma~\ref{lem:subpath} for a discussion why the $\tO(Mn^\omega)$ component from Lemma~\ref{lem:SSRP_to_subpath} can be omitted.} application of Lemmas~\ref{lem:SSRP_to_subpath} and~\ref{lem:subpath}.

\subparagraph*{Algorithm overview and the subpath problem.}
The algorithm first computes a shortest paths tree (from the source vertex $s$), and splits it into a subpolynomial number of subtrees. By using balanced separators, the subtrees can be roughly of the same size. Then, for each such subtree $T$, values $d_G(s, v, e)$ such that both vertex $v$ and edge $e$ belong to $T$ are deferred to a recursive call on a graph obtained from $G$ by carefully compressing its parts outside $T$. The only remaining interesting values $d_G(s, v, e)$ (i.e. such that they might be different from $d_G(s, v)$) are such that vertex $v$ belongs to subtree $T$ and edge $e$ lies on the path from $s$ to the root of $T$ in the shortest paths tree. The problem of computing those remaining values is called \emph{subpath problem}, which we now define formally.

\begin{definition}[Subpath problem]
Given an $n$-vertex directed graph $G$ with edge weights in $\{-M, \ldots, M\}$, a source vertex $s$, and a tree $T$ which is a subtree of a shortest paths tree from $s$, compute $d_G(s, v, e)$ for every $v \in T$ and every $e$ on the path from $s$ to the root $t$ of $T$ in the shortest paths tree.
\end{definition}

Using the ideas outlined above, Grandoni and Vassilevska Williams formally reduce SSRP to the subpath problem.

\begin{lemma}[Lemma~5.1 in~\cite{GVW2019}]
\label{lem:SSRP_to_subpath}
Given an algorithm that solves the subpath problem in time $\tO(M^\alpha n^\beta)$, with high probability, for constants $\beta \ge \alpha + 1 \ge 1$, there is an algorithm that solves SSRP in time $\tO(Mn^\omega + M^\alpha n^\beta)$, with high probability.
\end{lemma}

\subparagraph*{Jumping paths and departing paths.}
We proceed to show how to solve the subpath problem.
Let $P=(s=s_1 \to s_2 \to \cdots \to s_{|P|}=t)$ be the $s$-$t$ path in the shortest paths tree. A replacement path witnessing $d_G(s,v,e)$ has to depart from $P$ somewhere before $e$ and then can either (1) join $P$ back somewhere after $e$, and thus reach $v$ through $t$, or (2) never use any other edge of $P$ after departing. Paths of the first type are called \emph{jumping paths}, and of the second type -- \emph{departing paths}. Grandoni and Vassilevska Williams~\cite{GVW2019} use the fact that jumping paths can be found by solving the \emph{$s$-$t$ replacement paths} problem, i.e.~computing all $d_G(s,t,e)
$'s for fixed $t$, which can be computed in $\tO(Mn^\omega)$ time (see Lemma~\ref{lem:rp}). We just follow their approach in that regard.

\begin{lemma}[Theorem~1.1 in~\cite{GVW2019}]\label{lem:rp}
There is a randomized algorithm that solves $s$-$t$ replacement paths problem in $\tO(Mn^\omega)$ time, with high probability.
\end{lemma}

\subparagraph*{Improved algorithm for departing paths.}
Let $\widetilde{G}$ denote the graph obtained from $G$ by removing all edges on path $P$. Note that the length of a shortest departing replacement path to $v$ avoiding $e=(s_i,s_{i+1})$ equals to $\min_{j \le i} d_G(s,s_j) + d_{\widetilde{G}}(s_j,v)$. Grandoni and Vassilevska Williams~\cite{GVW2019} simply feed $\widetilde{G}$ to Zwick's APSP algorithm~\cite{Zwick02}, running in time $\tO(M^{\frac{1}{4-\omega}}n^{2+\frac{1}{4-\omega}})$, to compute all $d_{\widetilde{G}}(s_j,v)$'s. We take a different approach and employ our truly subcubic algorithm for Monotone Min-Plus Product. We remark that any truly subcubic algorithm would yield an improvement.

Let $\zeta \in [0,1]$ be a parameter to be determined later. We say that a departing replacement path is \emph{hop-long} if it visits at least $n^\zeta$ nodes after departing $P$, otherwise it is \emph{hop-short}. We handle the two types of paths separately.

\subparagraph*{Hop-short paths.} To find hop-short paths we use a modification of Zwick's APSP algorithm~\cite{Zwick02}, already described in~\cite{GVW2019}. Zwick's algorithm consists of $O(\log n)$ iterations, and in the $i$-th iteration it computes the shortest paths which use at most $(3/2)^i$ nodes. By running only first few iterations we can compute all hop-short shortest paths in time $\tO(Mn^{\zeta+\omega(1,1-\zeta,1)})$,
which is faster than it would take to compute all shortest paths (given that $\zeta$ is small enough). For a formal proof of this statement we refer to~\cite{GVW2019}.

\begin{lemma}[Corollary~3.1 in~\cite{GVW2019}]\label{lem:hopshort}
The distances between all pairs of nodes that have shortest paths on at most $n^\zeta$ nodes can be computed in time $\tO(Mn^{\zeta+\omega(1,1-\zeta,1)})$, with high probability.
\end{lemma}

\subparagraph*{Hop-long paths.} To find hop-long paths, first we sample (with replacement) $c \cdot n^{1-\zeta} \ln n$ nodes, for a large enough constant $c$. Let $B \subseteq V$ denote the set of sampled nodes. For the sake of analysis let us fix a set $\mathcal{S}$ of shortest hop-long departing replacement for all nodes $v \in T$ and all edges $e \in P$. 
When there is more than one such path of the smallest length for a given pair $(v,e)$, we choose an arbitrary one. Note that for paths in $\mathcal{S}$, we only include the portions after they depart $P$ so that they only contain edges in $\widetilde{G}$. Since the definition of hop-long paths only concerns the length of the part of a path after it departs $P$, all paths in $\mathcal{S}$ have length at least $n^{\zeta}$. By a standard proof, with high probability, every path in $\mathcal{S}$ contains a node from $B$ which lies in that path's middle third part.

% \begin{claim}
% With high probability, every path in $\mathcal{S}$ contains a node from $B$ which lies in that path's middle third part.
% \end{claim}

% \begin{proof}
% Fix a path in $\mathcal{S}$. It has at least $ \frac{1}{3} n^\zeta$ nodes in its middle third part. A single sample hits one of them with probability (at least) $\frac{1}{3} n^{\zeta - 1} = 1 / (3n^{1-\zeta})$. The probability that none of the independent $c \cdot n^{1-\zeta} \ln n$ samples hits one of the desired nodes is thus at most $\left(1 - 1 / (3n^{1-\zeta})\right) ^ {c \cdot n^{1-\zeta} \ln n} \le e^{-\frac{1}{3}c\ln n} = 1/n^{\frac{1}{3}c}$. By a union bound over at most $n^2$ paths, all of them get hit with probability at least $1-1/n^{\frac{1}{3}c - 2}$, as desired.
% \end{proof}

Then we construct Yuster-Zwick distance oracle for graph $\widetilde{G}$ (see Lemma~\ref{lem:oracle} below) and use it compute all $\tO(n \cdot n^{1-\zeta})$ shortest paths to and from $B$ using at least $(1/3) \cdot n^\zeta$ nodes. In total it takes time 
$\tO(Mn^\omega + n^{3-2\zeta}).$

\begin{lemma}[implicit in~\cite{YZ05}, Lemma~2.3 in~\cite{GVW2019}]\label{lem:oracle}
Given an $n$-vertex directed graph $G$, with edge weights in $\{-M, \ldots, M\}$, one can compute in $\tO(Mn^\omega)$ time an $n \times n$ matrix $D$, so that the $(i,j)$-th entry of the min-plus product $D \star D$ is the distance from node $i$ to $j$ in $G$. Furthermore, by the properties of $D$, the length of a shortest $i \leadsto j$ path containing at least $n^\zeta$ nodes can be computed in $\tO(n^{1-\zeta})$ extra time, with high probability.
\end{lemma}

Let $d^{YZ}(\cdot,\cdot)$ denote such computed distances. The length of a shortest hop-long departing replacement path to $v$ avoiding $e=(s_i,s_{i+1})$ equals  $\min_{j \le i} \min_{b \in B} d_G(s,s_j) + d^{YZ}(s_j, b) + d^{YZ}(b, v).$

We create two matrices $A$ and $B$, of dimensions at most $n\times|B|$ and $|B|\times n$, respectively, such that
\[
A_{i,b} =  \min_{j \le i} d_G(s,s_j) + d^{YZ}(s_j, b),
\quad\text{and}\quad
B_{b,v} = d^{YZ}(b,v).
\]
We need to compute $A \star B$. Note that $A_{i+1,b} \le A_{i,b}$, i.e.~columns of $A$ are monotone.
Moreover, finite entries of $A$ are of absolute value at most $2\cdot nM$, so we can compute $A \star B = (B^T\star A^T)^T$ in time
$\tO(n^{m(1,1-\zeta,1,1+\log_n M)}).$

\subparagraph*{Wrap-up and runtime analysis.}
% Now we can sum bounds from Lemma~\ref{lem:rp}, Lemma~\ref{lem:hopshort}, Equation~\eqref{eqn:timeoracle}, and Equation~\eqref{eqn:timeproduct},

Now we can sum up the running time and then balance the terms. 
% The proof of the following lemma is deferred to the \fullv.% appendix. 

\begin{lemma}\label{lem:subpath}
There is a randomized algorithm that solves the subpath problem in a directed $n$-vertex graph with edge weights in $\{-M, \ldots, M\}$ in $\tO(M^{\frac{5}{17-4\omega}}n^{\frac{36-7\omega}{17-4\omega}})$ time, with high probability. Using fast rectangular matrix multiplication improves the running time to $O(M^{0.8043}n^{2.4957})$.
\end{lemma}

\begin{proof}
Let $M=n^\mu$. By summing bounds from Lemma~\ref{lem:rp}, Lemma~\ref{lem:hopshort}, 
and the hop-long path part (Yuster-Zwick distance oracle, and a monotone min-plus product),
we get that the running time is
\[\tO\left(n^{\mu + \omega} + n^{\mu + \zeta + \omega(1,1-\zeta,1)} + n^{3 - 2 \zeta} + n^{m(1,1-\zeta,1,1+\mu)}\right).\]
Let us skip for a moment the $\tO(n^{\mu + \omega})$ component, the only one which does not depend on the parameter $\zeta$, and let us optimize that parameter, so the exponent at $n$ equals to
\[\min_{\zeta\in[0,1]} \max\{
  \mu + \zeta + \omega(1,1-\zeta,1),\ 
  3 - 2 \zeta,\ 
  m(1,1-\zeta,1,1+\mu)\},\]
which, by Theorem~\ref{thm:mmp-main}, is bounded by
\[\min_{\substack{\zeta\in[0,1]\\\theta\in[0,1+\mu]}} \max\{
  \mu + \zeta + \omega(1,1-\zeta,1),\ 
  3 - 2 \zeta,\ 
  3 + \mu - \theta - \zeta,\ 
  \frac{1}{2} (3 - \zeta + g(1,1-\zeta,1,\theta))\},\]
which, by Corollary~\ref{cor:g}, is bounded by
\[\min_{\substack{\zeta\in[0,1]\\\theta\in[0,1+\mu]}} \max\{
  \mu + \zeta + \omega(1,1-\zeta,1),\ 
  3 - 2 \zeta,\ 
  3 + \mu - \theta - \zeta,\ 
  \frac{1}{4}(9 - 3 \zeta + \omega(1,1-\zeta,1) + \theta)\}.\]

For the way we set $\zeta$, the $3-2\zeta$ term won't be dominant in the range of $\mu$ we are interested in. See the last paragraph of this proof for a discussion. 

If we want to express the running time as a function of square matrix multiplication exponent $\omega$, we need to implement rectangular matrix multiplication by cutting rectangular matrices into square matrices, which lets us bound $\omega(1,1-\zeta,1) \le \omega \cdot (1-\zeta) + 2\zeta.$ By setting $\zeta = 4\cdot \frac{3-\omega-\mu}{17-4\omega}$ and $\theta=(3-\omega)-(4-\omega)\cdot \zeta$, we get running time
\begin{equation}\label{eqn:ssrp_time_omega}
\tO(M^{\frac{5}{17-4\omega}}n^{\frac{36-7\omega}{17-4\omega}}).
\end{equation}

Alternatively, we can  set $\zeta=0.4035-0.6434\cdot\mu$ and $\theta=0.1009+0.8391\cdot\mu$
and use the best known bounds for rectangular matrix multiplication~\cite{LU18} to get a faster running time. We first use~\cite{LU18} to bound $\omega(1, 0.5965, 1) \le 2.0922$. Then since $0.5965 \le 1-\zeta \le 1$, we can use convexity (see e.g., \cite{lotti1983asymptotic}) to bound $\omega(1, 1-\zeta, 1)$ by $2.0922 \cdot \frac{1-(1-\zeta)}{1-0.5965} + 2.3729 \cdot \frac{(1-\zeta) - 0.5965}{1-0.5965} \le 2.0922 + 0.4476 \cdot \mu$.  It is then easy to see that our algorithm runs in time
\begin{equation}\label{eqn:ssrp_time_num}
O(M^{0.8043}n^{2.4957}).    
\end{equation}

To end the proof we remark that SSRP, and thus also the subpath problem, can be solved by a trivial algorithm in $\tO(n^3)$ time, and therefore we have to prove the two bounds of Equations~\eqref{eqn:ssrp_time_omega} and~\eqref{eqn:ssrp_time_num} only in the range of $\mu$ in which these bounds are subcubic, i.e.~$\mu \le 3-\omega$. In this regime the initially omitted $\tO(Mn^\omega)$ component and the $\tO(n^{3-2\zeta})$ component are not dominant, and the above values of $\zeta$ are non-negative, as required.
\end{proof}

\subsection{Lower Bound}
In this section, we prove our conditional lower bound for SSRP.

\ThmSSRPLowerBound*

We first reduce Bounded-Difference Min-Plus Product to a problem called Ham-APSP, then we further reduce Ham-APSP to SSRP.

\begin{problem}[Ham-APSP]
Given a directed unweighted graph $G$ with vertex set $\{v_1, \ldots, v_n\}$ and a Hamiltonian path $v_1 \rightarrow \cdots \rightarrow v_n$ of $G$, compute all pairs shortest path distances in $G$.
\end{problem}

The key idea in the following reduction was used by Chan et al.~\cite{CVX} for a reduction from Min-Plus Product with small integer weights to  unweighted directed APSP.

\begin{lemma}
\label{lem:min_plus_to_ham}
If there exists a $T(n)$ time algorithm for Ham-APSP in a graph with $n$ vertices, then there exists an $O(T(n) \sqrt{n})$ time algorithm for Bounded-Difference Min-Plus Product of $n \times n$ matrices.
\end{lemma}

\begin{proof}

Given two $n \times n$ bounded-difference matrices $A$ and $B$, we first split the columns of $A$ and rows of $B$ to $O(\sqrt{n})$ pieces. For each pair of pieces, we need to compute the min-plus product of an $n \times \sqrt{n}$ bounded-difference matrix $X$ and a $\sqrt{n} \times n$ bounded-difference matrix $Y$. We will use a single call of the assumed $T(n)$ time algorithm for Ham-APSP to compute the min-plus product between each pair of pieces, yielding an $O(T(n) \sqrt{n})$ overall running time.

We create a new matrix $X'$ such that $X'_{i, k} = X_{i, k} - X_{i, 1}$.
Since $|X_{i,k}-X_{i,k+1}| \le 1$ for any $i, k$,  all entries of $X'$ are bounded by $\sqrt{n}$. We can create $Y'$ similarly by setting $Y'_{k, j} = Y_{k, j} - Y_{1, j}$ so that all entries of $Y'$ are bounded by $\sqrt{n}$ as well. We will later use the Ham-APSP algorithm to compute $X' \star Y'$, which immediately gives $X \star Y$ via the relation $(X \star Y)_{i, j} = (X' \star Y')_{i,j}+X_{i,1}+Y_{1, j}$.

In \cite{CVX}, Min-Plus Product of an $n \times \sqrt{n}$ and a $\sqrt{n} \times n$ matrices with weights up to $\sqrt{n}$ is reduced to unweighted directed APSP on $n$-node graphs. Here is a description of that reduction. We create a vertex set $I$ of $n$ vertices $\{a_1, \ldots, a_n\}$ and a vertex set $J$ of $n$ vertices $\{b_1, \ldots, b_n\}$. We also create $\sqrt{n}$ paths $p_1, \ldots, p_{\sqrt{n}}$ each of length $2\sqrt{n}$. From each $a_i$ to $p_k$, we add a directed edge from $a_i$ to the $(\sqrt{n} - X'_{i, k})$-th node on $p_k$; similarly, from each $p_k$ to $b_j$, we add a directed edge from the $(\sqrt{n} + Y'_{k, j})$-th node on $p_k$ to $b_j$. Then we can see that the distance from $a_i$ to $b_j$ equals $(X' \star Y')_{i, j}+2$.

In order to have a Hamiltonian path in the graph, we need to add two types of additional paths. 

For the first type, we add  paths of length $2$ from $a_i$ to $a_{i+1}$ and from $b_{i}$ to $b_{i+1}$ for every $1 \le i < n$. Clearly, we only add $O(n)$ vertices and $O(n)$ edges. 
Now consider the shortest path from $a_i$ to $b_j$ for some $i, j$. The shortest path has the option to go to some $a_{i'}$ for $i' \ge i$, then choose some path $p_k$ in the middle, then go to $b_{j'}$ for $j' \le j$, and finally reach $b_j$. The cost of this path would be $2(i'-i) + 1 + X'_{i',k}+Y'_{k,j'}+1 + 2(j-j')$. Because  $X$ and $Y$ have bound differences, we have that $2(i'-i) + X'_{i',k} \ge X'_{i, k}$ and $Y'_{k,j'}+2(j-j') \ge Y'_{k, j}$. Therefore, in one of the shortest paths from $a_i$ to $b_j$, we have $i'=i$ and $j'=j$. Thus, the distance from $a_i$ to $b_j$ is exactly $2+ (X' \star Y')_{i, j}$, so we can recover $X' \star Y'$ by computing all the pairwise distances.

For the second type of paths, we add $O(\sqrt{n})$ paths of lengths $3\sqrt{n}$ to connect $I$, $J$ and each $p_k$, as shown in Figure~\ref{fig:min_plus_to_ham}. The total number of vertices and number of edges added are both $O(n)$. 
Since all distances we care about are at most $2\sqrt{n}+O(1)$, adding those paths won't affect these distances.

This graph now has a Hamiltonian path: we can travel from $a_1$ to $a_n$ via the first type of paths. Then we use the second type of paths to travel from $a_n$ to the beginning of $p_1$ and then we can easily travel to the end of $p_1$ by using edges of $p_1$. Similarly, we can go through all vertices in $p_2, \ldots, p_{\sqrt{n}}$. Finally, we travel from the end of $p_{\sqrt{n}}$ to $b_1$ via the second type of paths, and then use the first type of paths to travel to $b_n$. 
\end{proof}

\begin{figure}
    \centering
        \begin{tikzpicture}
		
	   \node at(-5, 3)  [circle,fill,inner sep=2pt,label=left:$a_1$] (a1){};
	   \node at(-5, 1)  [circle,fill,inner sep=2pt,label=left:$a_2$] (a2){};
	   \node at(-5, -3)  [circle,fill,inner sep=2pt,label=left:$a_n$] (an){};
	   \path (a2) -- node[auto=false]{\vdots} (an);

        \node[ellipse, draw,align=left,minimum width = 2.5cm,minimum height = 8cm,label={[shift={(0,0)}]$I$}] at (-5, 0) (I) {};

	   \node at(5, 3)  [circle,fill,inner sep=2pt,label=right:$b_1$] (b1){};
	   \node at(5, 1)  [circle,fill,inner sep=2pt,label=right:$b_2$] (b2){};
	   \node at(5, -3)  [circle,fill,inner sep=2pt,label=right:$b_n$] (bn){};
	   \path (b2) -- node[auto=false]{\vdots} (bn);
        \node[ellipse, draw,align=left,minimum width = 2.5cm,minimum height = 8cm,label={[shift={(0,0)}]$J$}] at (5, 0) (I) {};

	    \node at(-3, 2)  [] (p1){$p_1$};
        \node at(-2, 2)  [circle,fill,inner sep=1pt] (p11){};
        \node at(-1, 2)  [circle,fill,inner sep=1pt] (p12){};
        \node at(0, 2)  [circle,fill,inner sep=1pt] (p13){};
        \node at(1, 2)  [circle,fill,inner sep=1pt] (p14){};
        \node at(2, 2)  [circle,fill,inner sep=1pt] (p15){};
        \node at(3, 2)  [circle,fill,inner sep=1pt] (p16){};
         \draw[->, ] (p11) to[] node[] {} (p12);
         \draw[->, ] (p12) to[] node[] {} (p13);
          \path (p13) -- node[auto=false]{\ldots} (p14);
         \draw[->, ] (p14) to[] node[] {} (p15);
          \draw[->, ] (p15) to[] node[] {} (p16);

         \draw[->, bend left] (a1) to[] node[] {} (p12);
          \draw[->, bend left] (p14) to[] node[] {} (b1);
          \draw [] (-1, 2.1) -- (-1, 2.5);
          \draw [] (1, 2.1) -- (1, 2.5);
           \draw[<->] (-1, 2.3) -- (1, 2.3) node[midway, above]{\tiny{$X'_{1,1}+Y'_{1,1}-1$}};

		\node at(-3, 0.5)  [] (p1){$p_2$};
        \node at(-2, 0.5)  [circle,fill,inner sep=1pt] (p21){};
        \node at(-1, 0.5)  [circle,fill,inner sep=1pt] (p22){};
        \node at(0, 0.5)  [circle,fill,inner sep=1pt] (p23){};
        \node at(1, 0.5)  [circle,fill,inner sep=1pt] (p24){};
        \node at(2, 0.5)  [circle,fill,inner sep=1pt] (p25){};
        \node at(3, 0.5)  [circle,fill,inner sep=1pt] (p26){};
         \draw[->, ] (p21) to[] node[] {} (p22);
         \draw[->, ] (p22) to[] node[] {} (p23);
          \path (p23) -- node[auto=false]{\ldots} (p24);
         \draw[->, ] (p24) to[] node[] {} (p25);
          \draw[->, ] (p25) to[] node[] {} (p26);
		
		\node at(-3, -2)  [] (p1){$p_{\sqrt{n}}$};
        \node at(-2, -2)  [circle,fill,inner sep=1pt] (pn1){};
        \node at(-1, -2)  [circle,fill,inner sep=1pt] (pn2){};
        \node at(0, -2)  [circle,fill,inner sep=1pt] (pn3){};
        \node at(1, -2)  [circle,fill,inner sep=1pt] (pn4){};
        \node at(2, -2)  [circle,fill,inner sep=1pt] (pn5){};
        \node at(3, -2)  [circle,fill,inner sep=1pt] (pn6){};
         \draw[->, ] (pn1) to[] node[] {} (pn2);
         \draw[->, ] (pn2) to[] node[] {} (pn3);
          \path (pn3) -- node[auto=false]{\ldots} (pn4);
         \draw[->, ] (pn4) to[] node[] {} (pn5);
          \draw[->, ] (pn5) to[] node[] {} (pn6);

          \path (p23) -- node[auto=false]{\vdots} (pn3);

		\node at(-5, 2)  [circle,fill=green,inner sep=1pt] (a1.5){};
		 \draw[->, green] (a1) to[] node[] {} (-5, 2) ;
		\draw[->, green] (-5, 2) to[] node[] {} (a2) ;
		
		\node at(-5, 0)  [circle,fill=green,inner sep=1pt] (){};
		 \draw[->, green] (a2) to[] node[] {} (-5,0) ;		 
		\draw[->, green] (-5,0) to[] node[] {} (-5,-0.9) ;
		
		\node at(-5, -2.2)  [circle,fill=green,inner sep=1pt] (){};
		 \draw[->, green] (-5,-2.2) to[] node[] {} (an) ;		 
		\draw[->, green] (-5, -1.4) to[] node[] {} (-5,-2.2) ;

		\node at(5, 2)  [circle,fill=green,inner sep=1pt] (b1.5){};
		 \draw[->, green] (b1) to[] node[] {} (5, 2) ;
		\draw[->, green] (5, 2) to[] node[] {} (b2) ;
		
		\node at(5, 0)  [circle,fill=green,inner sep=1pt] (){};
		 \draw[->, green] (b2) to[] node[] {} (5,0) ;		 
		\draw[->, green] (5,0) to[] node[] {} (5,-0.9) ;
		
		\node at(5, -2.2)  [circle,fill=green,inner sep=1pt] (){};
		 \draw[->, green] (5,-2.2) to[] node[] {} (bn) ;		 
		\draw[->, green] (5, -1.4) to[] node[] {} (5,-2.2) ;

		\tikzset{
my dash/.style={->, blue, dash pattern=on 15pt off 2pt,postaction={decorate,
    decoration={markings,
    mark=between positions 15pt and 1 step 17pt with {\arrow{>};}}},
                }
         }
         
      \draw[my dash] (an) to[] node[] {} (p11) ;
      \draw[blue, line width=3pt, line cap=round, dash pattern=on 0pt off 17pt, dash phase=2pt] (an) to[] node[] {} (p11) ;
      
       \draw[my dash] (p16) to[] node[] {} (p21) ;
      \draw[blue, line width=3pt, line cap=round, dash pattern=on 0pt off 17pt, dash phase=2pt] (p16) to[] node[] {} (p21) ;

      \draw[my dash] (p26) to[] node[] {} (-2, -0.75) ;
      \draw[blue, line width=3pt, line cap=round, dash pattern=on 0pt off 17pt, dash phase=2pt] (p26) to[] node[] {} (-2, -0.75);

       \draw[my dash] (3, -0.75) to[] node[] {} (pn1) ;
      \draw[blue, line width=3pt, line cap=round, dash pattern=on 0pt off 17pt, dash phase=2pt] (pn1) to[] node[] {} (3, -0.75);
      
      \draw[my dash] (pn6) to[] node[] {} (b1) ;
      \draw[blue, line width=3pt, line cap=round, dash pattern=on 0pt off 17pt, dash phase=2pt] (pn6) to[] node[] {} (b1);
    \end{tikzpicture}
    \caption{Reduction from Bounded-Difference Min-Plus Product to Ham-APSP (Lemma~\ref{lem:min_plus_to_ham}). Most edges between the parts $I, J$ and the middle paths $p_1, \ldots, p_{\sqrt{n}}$ are omitted for clarity. The green portions are the first type paths, and the blue portions are the second type paths. }
    \label{fig:min_plus_to_ham}
\end{figure}
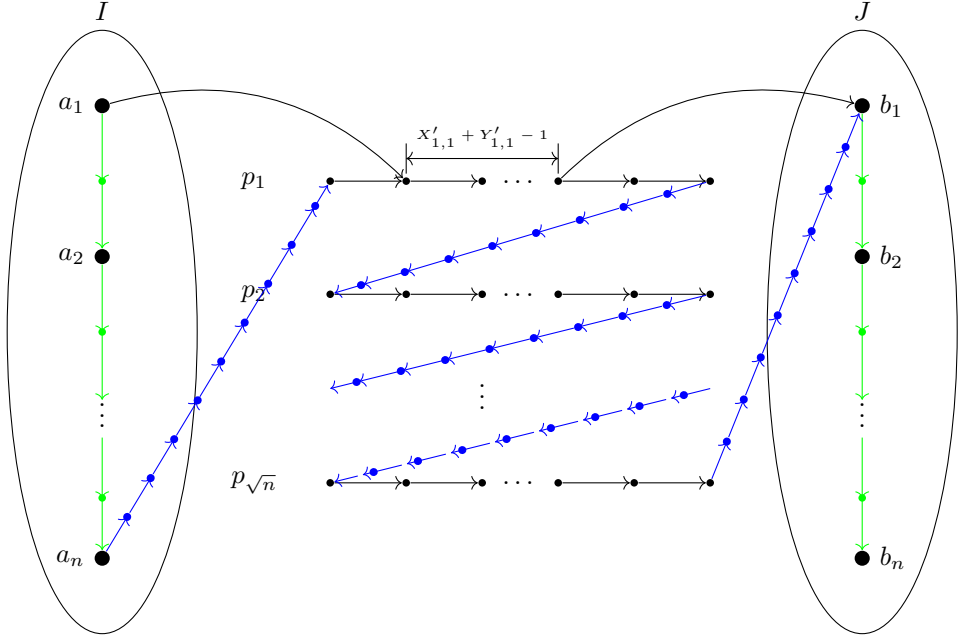

% \begin{figure}[H]
%     \centering
%     \begin{tikzpicture}
%         \node at(0, 0)  [circle,fill,inner sep=2pt,label=above:$v_{n+1}'$] (s){};
%         \node at(2, 0)  [circle,fill,inner sep=2pt,label=above:$v_{n}'$] (vn'){};
%         \node at(4, 0)  [circle,fill,inner sep=2pt,label=above:$v_{n-1}'$] (vn-1'){};
%         \node at(5.5, 0)  [circle,fill,inner sep=2pt,label=above:$v_{1}'$] (v1'){};
%         \node at(7.5, 0)  [circle,fill,inner sep=2pt,label=above:$v_{0}'$] (v0){};

%         \draw[->, ] (s) to[] node[label=above:$-1$] {} (vn');
%         \draw[->, ] (vn') to[] node[label=above:$-1$] {} (vn-1');
%         \path (vn-1') -- node[auto=false]{\ldots} (v1');
%         \draw[->, ] (v1') to[] node[label=above:$-1$] {} (v0);

%         \node[ellipse, draw,align=left, xshift=9cm, yshift=-2.5cm,minimum width = 8cm,minimum height = 3cm,label={[shift={(0,-2.8)}]$G$}] at (-5, 0) (I) {};

%         \node at(3, -3)  [circle,fill,inner sep=2pt,label=left:$v_{n}$] (vn){};
%         \node at(4, -2)  [circle,fill,inner sep=2pt,label=left:$v_{n-1}$] (vn-1){};
%         \node at(5, -2.5)  [circle,fill,inner sep=2pt,label=right:$v_{1}$] (v1){};

%         \draw[->, ] (vn') to[] node[label=right:$0$] {} (vn);
%         \draw[->, ] (vn-1') to[] node[yshift=-8pt,label=right:$0$] {} (vn-1);
%         \draw[->, ] (v1') to[] node[yshift=-3pt,label=right:$0$] {} (v1);
%     \end{tikzpicture}
%     \caption{Reduction from Ham-APSP to SSRP with weights in $\{-1, 0, 1\}$ (Lemma~\ref{lem:hamapsp-to-ssrp}).}
%     \label{fig:ham_to_ssrp}
% \end{figure}

In the following lemma, we further reduce Ham-APSP to SSRP.

\begin{lemma}\label{lem:hamapsp-to-ssrp}
If there exists a $T(n)$ time algorithm for SSRP in a graph with $n$ vertices whose edge weights are in $\{-1, 0, 1\}$, then there exists an $O(T(n))$ time algorithm for Ham-APSP in a graph with $n$ vertices.
\end{lemma}

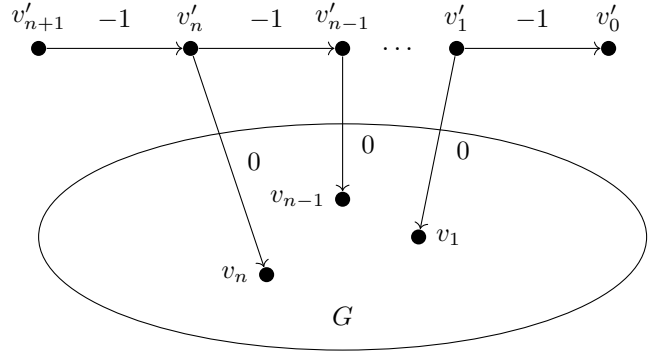
\begin{figure}
    \centering
    \begin{tikzpicture}
        \node at(0, 0)  [circle,fill,inner sep=2pt,label=above:$v_{n+1}'$] (s){};
        \node at(2, 0)  [circle,fill,inner sep=2pt,label=above:$v_{n}'$] (vn'){};
        \node at(4, 0)  [circle,fill,inner sep=2pt,label=above:$v_{n-1}'$] (vn-1'){};
        \node at(5.5, 0)  [circle,fill,inner sep=2pt,label=above:$v_{1}'$] (v1'){};
        \node at(7.5, 0)  [circle,fill,inner sep=2pt,label=above:$v_{0}'$] (v0){};

        \draw[->, ] (s) to[] node[label=above:$-1$] {} (vn');
        \draw[->, ] (vn') to[] node[label=above:$-1$] {} (vn-1');
        \path (vn-1') -- node[auto=false]{\ldots} (v1');
        \draw[->, ] (v1') to[] node[label=above:$-1$] {} (v0);

        \node[ellipse, draw,align=left, xshift=9cm, yshift=-2.5cm,minimum width = 8cm,minimum height = 3cm,label={[shift={(0,-2.8)}]$G$}] at (-5, 0) (I) {};

        \node at(3, -3)  [circle,fill,inner sep=2pt,label=left:$v_{n}$] (vn){};
        \node at(4, -2)  [circle,fill,inner sep=2pt,label=left:$v_{n-1}$] (vn-1){};
        \node at(5, -2.5)  [circle,fill,inner sep=2pt,label=right:$v_{1}$] (v1){};

        \draw[->, ] (vn') to[] node[label=right:$0$] {} (vn);
        \draw[->, ] (vn-1') to[] node[yshift=-8pt,label=right:$0$] {} (vn-1);
        \draw[->, ] (v1') to[] node[yshift=-3pt,label=right:$0$] {} (v1);
    \end{tikzpicture}
    \caption{Reduction from Ham-APSP to SSRP with weights in $\{-1, 0, 1\}$ (Lemma~\ref{lem:hamapsp-to-ssrp}).}
    \label{fig:ham_to_ssrp}
\end{figure}

\begin{proof}
Let $G$ be an instance of the Ham-APSP problem. We first create a graph $G'$  whose vertex set is $\{v_0', v_1', \ldots, v_n', v_{n+1}'\} \cup \{v_1, \ldots, v_n\}$. Then we add the following three types of edges to $G$, as depicted in Figure~\ref{fig:ham_to_ssrp}:
\begin{enumerate}
    \item We add an edge from $v_{i}'$ to $v_{i-1}'$ of weight $-1$ for every $1 \le i \le n+1$ ;
    \item We add an edge from $v_{i}'$ to $v_i$ of weight $0$ for every $1 \le i \le n$ ;
    \item We add an edge from $v_{i}$ to $v_j$ of weight $1$ for every $(v_i, v_j) \in E(G)$. This part essentially pastes a copy of $G$ to $G'$.
\end{enumerate}

If we cut the edge $(v_i', v_{i-1}')$ in the graph $G'$, then the shortest path from $v_{n+1}'$ to $v_j$ for some $1 \le j \le n$ must move from $v_{n+1}'$ to $v_k'$ for some $i \le k \le n$, then take the weight $0$ edge to $v_k$, and finally move to $v_j$ in the copy of graph $G$. Therefore, $d_{G'}(v_{n+1}', v_i, (v_i', v_{i-1}')) = \min_{i \le k \le n} (k-n-1)+d_G(v_k, v_j)$. We show that $\min_{i \le k \le n} (k-n-1)+d_G(v_k, v_j) = (i-n-1)+d_G(v_i, v_j)$. Clearly, $\min_{i \le k \le n} (k-n-1)+d_G(v_k, v_j) \le (i-n-1)+d_G(v_i, v_j)$ since the right hand side is one of the terms we are minimizing over.

To show the other direction, we fix an arbitrary $k \in [i, n]$. By triangle inequality, $d_G(v_i, v_k) + d_G(v_k, v_j) \ge d_G(v_i, v_j)$. Since $G$ has a Hamiltonian path $v_1 \rightarrow \cdots \rightarrow v_n$, it holds that $d_G(v_i, v_k) \le k-i$. Hence, \[(k-n-1)+d_G(v_k, v_j) \ge (k-n-1)+d_G(v_i, v_j) - d_G(v_i, v_k) \ge  (i-n-1)+d_G(v_i, v_j).\]

We have shown that $d_{G'}(v_{n+1}', v_i, (v_i', v_{i-1}')) = (i-n-1)+d_G(v_i, v_j)$. Thus, we can infer the pairwise distances in $G$ by querying the assumed $T(n)$ time SSRP algorithm on graph $G'$ since $d_G(v_i, v_j) = d_{G'}(v_{n+1}', v_i, (v_i', v_{i-1}'))- (i - n -1)$.
\end{proof}

\section{Range Mode}\label{sec:range-mode}
In this section, we show our improved algorithms for Batch Range Mode and Dynamic Range Mode. 

\subsection{Batch Range Mode}

\ThmBatch*

\begin{proof}
	Via a binary search, Sandlund and Xu \cite{SX20} showed that the Batch Range Mode problem can be reduced to finding the frequencies of the most frequent elements for all queries (with an $\tO(1)$ factor overhead),
	which is in turn reduced to Monotone Min-Plus Product in \cite[Theorem 6.1]{SX20}, leading to a deterministic \[\tO(n^{\min_{0\le \tau \le 1} \max\{ m(1-\tau, 1-\tau,1-\tau, \tau), 1+\tau\}})\]
	time algorithm for Batch Range Mode.

	Expanding the expression using Theorem \ref{thm:mmp-main} and \eqref{eqn:g-111-bound}, we have
	\begin{align}
		& \min_{0\le \tau \le 1} \max\{m(1-\tau, 1-\tau,1-\tau, \tau), 1+\tau\} \nonumber \\
		& = \min_{0\le \tau \le 1} \max\{(1-\tau) m(1,1,1,\frac{\tau}{1-\tau}), 1+\tau\} \nonumber \\
		& \le \min_{0\le \tau \le 1} \max\{(1-\tau) \min_{0\le \theta \le \frac{\tau}{1-\tau}} \max\{2+\frac{\tau}{1-\tau}-\theta, \frac 12(3 + g(1,1,1, \theta))\}, 1+\tau\} \nonumber \\
		& = \min_{\substack{0\le \tau \le 1\\ 0\le \theta \le \frac{\tau}{1-\tau}}} \max\{\tau+(1-\tau)(2-\theta), \frac{1-\tau}2 (3+g(1,1,1,\theta)), 1+\tau\} \nonumber \\
		& \le \min_{\substack{0\le \tau \le 1\\ 0\le \theta \le \frac{\tau}{1-\tau}\\ 0\le \delta \le 1}} \max\{\tau+(1-\tau)(2-\theta), \frac{1-\tau}2 (3+\max\{\omega(1,1,2-\delta)+\theta,2+\delta\}), 1+\tau\}. \label{eqn:range-mode}
	\end{align}
    
    By using $\omega(1, 1, 2 - \delta) \le 1-\delta + \omega$ and taking $\delta = \frac{4\omega - 3}{9}, \theta = \frac{3-\omega}{9}, \tau = \frac{6+\omega}{15+\omega}$, we can upper bound \eqref{eqn:range-mode} by $\frac{21+2\omega}{15+\omega}$.
	Using the upper bounds in \cite{LU18}, we find that when
	$\tau = 0.4804$, $\theta = 0.0754$, $\delta = 0.6984$, expression \eqref{eqn:range-mode} takes value $\le 1.4805$.
\end{proof}

\subsection{Dynamic Range Mode}
\ThmDynamic*

Our strategy is to improve \cite[Lemma 11]{SX20}. The following  Min-Plus-Query-Witness problem defined in \cite[Problem 7]{SX20} plays a key role in the algorithm for Dynamic Range Mode.
\begin{problem}[Min-Plus-Query-Witness problem]
	We are given two matrices $A$ and $B$ and are able to perform preprocessing before the first query.
	For each query, we are given two indices $i,j$ and a set $S$ of indices, and we must output
	an index $k^*\not \in S$ such that
	\[A_{i,k^*} + B_{k^*,j} = \min_{k\not \in S} \{A_{i,k} + B_{k,j}\}.\]
\end{problem}

We recall the following lemma from \cite{SX20}.

\begin{lemma}[{\cite[Lemma 9]{SX20}}]\label{lemma:mpqw-one-bounded}
	Let $\beta$ and $\theta$ be non-negative real numbers.
	The Min-Plus-Query-Witness problem where $A$ is an $n \times n^\beta$ matrix whose entries are in $\{-n^\theta, \ldots, n^\theta\} \cup \{\infty\}$ and $B$ is an $n^\beta \times n$ matrix
	can be solved with
	\begin{itemize}
		\item $\tO(n^{\theta + \omega(1, \beta, 1) + \beta -\sigma})$ preprocessing time,
		\item $\tO(|S| + n^{\sigma})$ worst-case time per query,
		\item and $\tO(n^{\max\{2+\beta+\theta, 1+2\beta\}-\sigma})$ space,
	\end{itemize}
	for every $0\le \sigma\le \beta$.
\end{lemma}

The following is our key lemma for the dynamic range mode algorithm. 

\begin{lemma}\label{lemma:mpqw-monotone}
	The Min-Plus-Query-Witness problem where $A$ is an $n \times n^\beta$ matrix, $B$ is an $n^\beta \times n$ monotone matrix with total range $O(n^{\beta+\eta})$, and the size of $S$ for each query is $O(n^\lambda)$,
	can be solved with
	\begin{itemize}
		\item $\tO(n^{\max\{1+\beta+\eta-\theta, \rho + \theta + \omega(1,\beta,1)+\beta-\sigma, 2+\beta-\rho\}})$ preprocessing time,
		\item $\tO(n^\lambda + n^{\rho + \sigma})$ worst-case time per query,
		\item and $\tO(n^{\max\{1+\beta+\eta-\theta,\rho + \max\{2+\beta+\theta, 1+2\beta\}-\sigma, 2+\beta-\rho\}})$ space,
	\end{itemize}
	for any constants $0\le \theta\le \eta$, $\rho \ge 0$ and $0\le \sigma \le \beta$.
\end{lemma}

\begin{proof}
	Preprocessing is in three steps, corresponding to the three phases of the algorithm for Theorem~\ref{thm:mmp-main}.
	\subparagraph*{Preprocessing Step 1.}
	This step is slightly more complicated than Phase 1 of Theorem~\ref{thm:mmp-main}.
	Let $0\le \theta \le \eta$ be a parameter to be chosen.
	Let $W = \lfloor n^\theta \rfloor$.
	Define two matrices $\ti A$ and $\ti B$ as $\ti A_{i,k} = \lfloor \frac {A_{i,k}}W\rfloor$
	and $\ti B_{k,j} = \lfloor \frac {B_{k,j}}W\rfloor$.
	
	We iterate through $j\in [n]$, and in each iteration $j$, maintain for each $i\in [n]$ the set 
	$L_{i,j} := \{(\ti A_{i,k} + \ti B_{k,j}, k) : k\in [n^\beta]\}$
	using a \textit{persistent} balanced BST.
	
	Using the maintained information,
	we can compute a matrix $\ti C'$,
	where each entry $\ti C'_{i,j}$ is
	defined as $W$ times the $(n^\lambda+1)$-th smallest element in $\{\ti A_{i,k} + \ti B_{k,j} : k\in [n^\beta]\}$.
	Furthermore, for each $(i,j)$ and for any $t$, we can enumerate the $t$ smallest elements in $L_{i,j}$ in $\tO(t)$ time.
	
	Time complexity and space complexity of this step are both $\tO(n^{\max\{2, 1+\beta+\eta-\theta\}})$.

	\subparagraph*{Preprocessing Step 2.}
	As in Phase 2 of Theorem~\ref{thm:mmp-main}, we sample $j^r \in [n]$ for $r\in [\tO(n^\rho)]$
	for some parameter $\rho$ to be chosen, and compute matrices $A^r$ and $B^r$ as
\begin{align*}
&    A^r_{i,k} = 
\left\{
\begin{array}{ll}
A_{i,k} + B_{k,j^r} - \ti C'_{i,j^r} & \text{if } |A_{i,k} + B_{k,j^r} - \ti C'_{i,j^r}| \le 3W \text{ and } A^{r^\p}_{i,k}=\infty\forall r^\p<r,\\
\infty & \text{otherwise,}
\end{array}
\right.
\\
&B^r_{k,j} = 
\left\{
\begin{array}{ll}
B_{k,j} - B_{k,j^r} & \text{if } B_{k, j^r} \ne \infty,\\
0 & \text{otherwise.}
\end{array}
\right.
\end{align*}

	Then for each $r$, we apply the data structure in Lemma~\ref{lemma:mpqw-one-bounded} to $A^r$ and $B^r$.
	Running time for this step is $\tO(n^{\rho+\theta + \omega(1, \beta, 1) + \beta -\sigma})$.
	Space complexity for this step is $\tO(n^{\rho + \max\{2+\beta+\theta, 1+2\beta\}-\sigma})$.

	We note that this part can be derandomized as well.

	\subparagraph*{Preprocessing Step 3.}
	For a pair $(i,k)$ if $A_{i,k}^r\ne \infty$ for some $r$, we call $(i,k)$ \emph{covered}; otherwise we call it \emph{uncovered}.
	We call a triple $(i,k,j)$ \emph{almost relevant}
	if $0\le \ti A_{i,k} + \ti B_{k,j} - \frac 1W \ti C'_{i,j} \le 1.$
	Note that for such triples, we must have
	$|A_{i,k} + B_{k,j} - \ti C'_{i,j}| \le 3W.$
	So the number of uncovered and almost relevant triples is $\tO(n^{2+\beta-\rho})$.
	
	We run an algorithm similar to Lemma~\ref{lemma:mmp-phase-3} to enumerate all uncovered and almost relevant triples $(i,k,j)$.
	Then for each $i, j\in [n]$, we use a balanced BST to store the set $T_{i,j} := \{(A_{i,k}+B_{k,j}, k) : \text{$(i,k,j)$ is uncovered and almost relevant}\}.$

	Running time for this step is $\tO(n^{\max\{2, 1+\beta+\eta-\theta, 2+\beta-\rho\}})$.
	Space complexity for this step is $\tO(n^{2+\beta-\rho})$.

	\subparagraph*{Query.}
	Now let us describe how to handle a query.
	Let $(S,i,j)$ be a query, and $k^*$ be an optimal index, i.e.,
	$A_{i,k^*} + B_{k^*,j} = \min_{k\not \in S} \{A_{i,k} + B_{k,j}\}.$
	There are three cases.

	\subparagraph*{Case 1: $(i,k^*)$ is covered.}
	For each $r$, we query the data structure (Lemma~\ref{lemma:mpqw-one-bounded}) for $A^r$ and $B^r$ with
	$(S_r, i, j)$
	where $S_r = S \cap \{k : A^r_{i,k}\ne \infty\}$.
	Because finite entries of $A^r$ are disjoint for different $r$,
	we have $\sum_r |S_r| = |S|$.
	So the total query time for this case is $\tO(|S| + n^{\rho+\sigma})$.

	\subparagraph*{Case 2: $(i,k^*)$ is uncovered and almost relevant.}
	In this case $k^*$ must be among the $(|S|+1)$ smallest elements in $T_{i,j}$.
	We deal with this case by enumerating the $(|S|+1)$ smallest elements in $T_{i,j}$.
	The query time for this case is $\tO(|S|)$.
	
	\subparagraph*{Case 3: $(i,k^*)$ is uncovered and not almost relevant.}
	Note that by definition of $\ti C'_{i,j}$, in this case we must have $\ti A_{i,k^*} + \ti B_{k^*,j} - \frac 1W \ti C'_{i,j} \le -1$.
	Therefore $k^*$ must be among the $(n^\lambda+1)$ smallest elements in $L_{i,j}$.
	We deal with this case by enumerating the $(n^\lambda+1)$ smallest elements in $L_{i,j}$.
	The query time for this case is $\tO(n^\lambda)$.

	\subparagraph*{Summary.}
	Total preprocessing time is $\tO(n^{\max\{1+\beta+\eta-\theta,\rho+\theta + \omega(1, \beta, 1) + \beta -\sigma, 2+\beta-\rho\}}).$
	Space complexity is $\tO(n^{\max\{1+\beta+\eta-\theta,\rho + \max\{2+\beta+\theta, 1+2\beta\}-\sigma, 2+\beta-\rho\}}).$
	Each query costs $\tO(n^\lambda + n^{\rho+\sigma})$ time.
\end{proof}

\begin{proof}[Proof of Theorem~\ref{thm:dynamic-range-mode}]
	The algorithm is exactly the same as in \cite{SX20}, except for replacing \cite[Lemma 11]{SX20}
	with Lemma~\ref{lemma:mpqw-monotone}.
	We skip most of the algorithm description and focus on analyzing the time and space complexity.

	For the algorithm we need to choose three constants $t_1, t_2, t_3\in [0, 1]$.
	The algorithm has three parts.

	\subparagraph*{Part 1: Infrequent values.}
	In the first part, we handle values that appear at most $n^{1-t_1}$ times.
	By maintaining $n^{1-t_1}$ balanced BSTs, this part can be done in $\tO(n^{2-2t_1})$ time per operation.
	Space complexity is $\tO(n^{2-t_1})$.

	\subparagraph*{Part 2: Newly modified values.}
	We maintain a data structure holding frequent values that rebuilds every $n^{t_2}$ operations.
	The data structure is discussed in Part 3.
	In Part 2, we deal with values that have been modified in the last $n^{t_2}$ operations.
	This part can be done in $\tO(n^{t_2})$ time and $\tO(n^{t_2})$ space.

	\subparagraph*{Part 3: Data structure.}
	In Part 3, we build the data structure.
	For this part, we call Lemma~\ref{lemma:mpqw-monotone} where
	$A$ is an $n^{1-t_3}\times n^{t_1}$ matrix,
	$B$ is an $n^{t_1}\times n^{1-t_3}$ monotone matrix with total range $O(n)$,
	and in the queries we have $|S| = O(n^{t_2})$.
	The dimensions $\alpha=\gamma=1-t_3$ come from choosing $n^{1-t_3}$ evenly spaced points in $[n]$, and $\beta=t_1$ comes from the $O(n^{t_1})$ values which appear more than $n^{1-t_1}$ times.
	In a query we pick the maximum interval whose endpoints are chosen points inside the query interval, and $S$ is the set of values modified in the last $n^{t_2}$ operations.
	
	Rebuilding costs
	$\tO(n^{(1-t_3)\max\{1+\beta+\eta-\theta, \rho + \theta + \omega(1,\beta,1)+\beta-\sigma, 2+\beta-\rho\}-t_2})$
	time per operation\footnote{The rebuilding is done every $n^{t_2}$ operations, so a priori we get an amortized time complexity.
	However, as pointed out in \cite{SX20}, we can use the global rebuilding technique of \cite{overmars1987design}
	to achieve a worst-case time bound.},
	where $\beta = \frac {t_1}{1-t_3}$, $\eta = \frac{1-t_1}{1-t_3}$,
	and $0\le \theta\le \eta$, $\rho \ge 0$, $0\le \sigma \le \beta$ are constants to be chosen.
	One query in Lemma~\ref{lemma:mpqw-monotone} costs $\tO(n^{t_2} + n^{(1-t_3)(\rho + \sigma)})$ time.
% 	\yinzhan{We never define ``segments'' so maybe we should rephrase the following sentence.}
	Besides the above, we also need $\tO(n^{t_3})$ time per operation to deal with elements not covered by the maximum interval of chosen points inside the query interval.
	
	Space cost of the data structure is $\tO(n^{(1-t_3)\max\{\rho + \max\{2+\beta+\theta, 1+2\beta\}-\sigma, 2+\beta-\rho\}}).$

	\subparagraph*{Summary.}
	Running time per operation is
	\[\tO(n^{2-2t_1} + n^{t_2} + n^{(1-t_3)\max\{1+\beta+\eta-\theta, \rho + \theta + \omega(1,\beta,1)+\beta-\sigma, 2+\beta-\rho\}-t_2}
	+ n^{(1-t_3)(\rho+\sigma)} + n^{t_3}),\]
	subject to
	$t_1,t_2,t_3\in [0, 1]$,
	$\beta = \frac {t_1}{1-t_3}$, $\eta = \frac{1-t_1}{1-t_3}$,
	$0\le \theta\le \eta$, $\rho \ge 0$, $0\le \sigma \le \beta$.
	
	Space cost is
	$\tO(n^{2-t_1} + n^{t_2} + n^{(1-t_3)\max\{\rho + \max\{2+\beta+\theta, 1+2\beta\}-\sigma, 2+\beta-\rho\}}).$
    
    As an observation, in the optimum case we have $\beta \ge 1$, so we may use $\omega(1, \beta, 1) \le \omega + \beta - 1$. 
    As a result, we can set $t_1=\frac{\omega+21}{2\omega+30}, t_2 = t_3 = \frac{\omega+9}{\omega+15}, \theta = \frac{3-\omega}{6}, \rho = \frac{3-\omega}{4}$ and $\sigma = \frac{5\omega+9}{12}$ to upper bound the running time by $\tO(n^{\frac{\omega+9}{\omega+15}})$. The space complexity is dominated by Part 1, which is $\tO(n^{2-t_1}) = \tO(n^{\frac{3\omega+39}{2\omega + 30}})$.
    
    Taking $t_1=0.67385$, $t_2=t_3=0.6523$, $\theta = 0.1239$, $\rho=0.1859$, $\sigma=1.6902$ and using fast rectangular matrix multiplication
	we get $O(n^{0.6524})$ running time per operation.
	The space complexity is again dominated by Part 1, which is $\tO(n^{2-t_1}) = O(n^{1.3262})$.
\end{proof}

\bibliography{ref}

% \appendix

% \section{Proof of Lemma~\ref{lem:subpath}}
% \input{subpath_proof}

% \section{Derandomization for Theorem \ref{thm:mmp-main}}
% \input{derand}

\end{document}